\def\year{2020}\relax
\documentclass[letterpaper]{article} 
\usepackage{aaai20}  
\usepackage{times}  
\usepackage{helvet} 
\usepackage{courier}  
\usepackage[hyphens]{url}  
\usepackage{graphicx} 
\usepackage{graphicx}  
\usepackage{amsmath}
\usepackage{amsfonts} 
\usepackage{multirow}
\usepackage{booktabs} 
\usepackage{enumitem}
\usepackage{algorithm}
\usepackage{algorithmic}
\usepackage{bm}
\newtheorem{assumption}{Assumption}
\newtheorem{lemma}{Lemma}
\newtheorem{thm}{Theorem}
\newcommand{\citet}[1]{\citeauthor{#1}\shortcite{#1}}
\newenvironment{proof}{{\noindent \it Proof}}{\hfill $\square$\par}
\title{SetRank: A Setwise Bayesian Approach for Collaborative Ranking\\ from Implicit Feedback}
\author{Chao Wang\textsuperscript{\rm 1,2}, Hengshu Zhu\textsuperscript{\rm 2}$ ^* $, Chen Zhu\textsuperscript{\rm 2}, Chuan Qin\textsuperscript{\rm 1,2}, Hui Xiong\textsuperscript{\rm 1,2,3}\thanks{Hui Xiong and Hengshu Zhu are corresponding authors.}\\  
\textsuperscript{\rm 1}University of Science and Technology of China\\ 
\textsuperscript{\rm 2}Baidu Talent Intelligence Center, Baidu Inc.\\
\textsuperscript{\rm 3}Business Intelligence Lab, Baidu Research\\
wdyx2012@mail.ustc.edu.cn, zhuhengshu@baidu.com, \{zc3930155, chuanqin0426, xionghui\}@gmail.com 
}
\begin{document}

\maketitle

\begin{abstract}
The recent development of online recommender systems has a focus on collaborative ranking from implicit feedback, such as user clicks and purchases. Different from explicit ratings, which reflect graded user preferences, the implicit feedback only generates positive and unobserved labels. While considerable efforts have been made in this direction, the well-known pairwise and listwise approaches have still been limited by various challenges. Specifically, for the pairwise approaches, the assumption of independent pairwise preference is not always held in practice. Also, the listwise approaches cannot efficiently accommodate ``ties'' due to the precondition of the entire list permutation. To this end, in this paper, we propose a novel setwise Bayesian approach for collaborative ranking, namely SetRank, to inherently accommodate the characteristics of implicit feedback in recommender system. Specifically, SetRank aims at maximizing the posterior probability of novel setwise preference comparisons and can be implemented with matrix factorization and neural networks. Meanwhile, we also present the theoretical analysis of SetRank to show that the bound of excess risk can be proportional to $\sqrt{M/N}$, where $M$ and $N$ are the numbers of items and users, respectively. Finally, extensive experiments on four real-world datasets clearly validate the superiority of SetRank compared with various state-of-the-art baselines.	
\end{abstract}

\section{Introduction}\label{Introduction}
Recommender systems have been widely deployed in many popular online services for enhancing user experience and business revenue~\cite{wang2015collaborative,hao2019trans2vec,Hydra,zhu2018person}. As one representative task of personalized recommendation, collaborative ranking aims at providing a user-specific item ranking for users based on their preferences learned from historical feedback. Indeed, in real-world scenarios, most of the user feedback is implicit (e.g., clicks and purchases) but not explicit (e.g., 5-star ratings). Different from explicit ratings, the implicit feedback only contains positive and unobserved labels instead of graded user preferences, which brings new research challenges for building recommender systems~\cite{hsieh2015pu}. Therefore, collaborative ranking from implicit feedback has been attracting more and more attention in recent years~\cite{rendle2009bpr,shi2010list,huang2015listwise,xia2019learning}.

\begin{figure*}[t]
	\centering
	\includegraphics[width=0.7\textwidth]{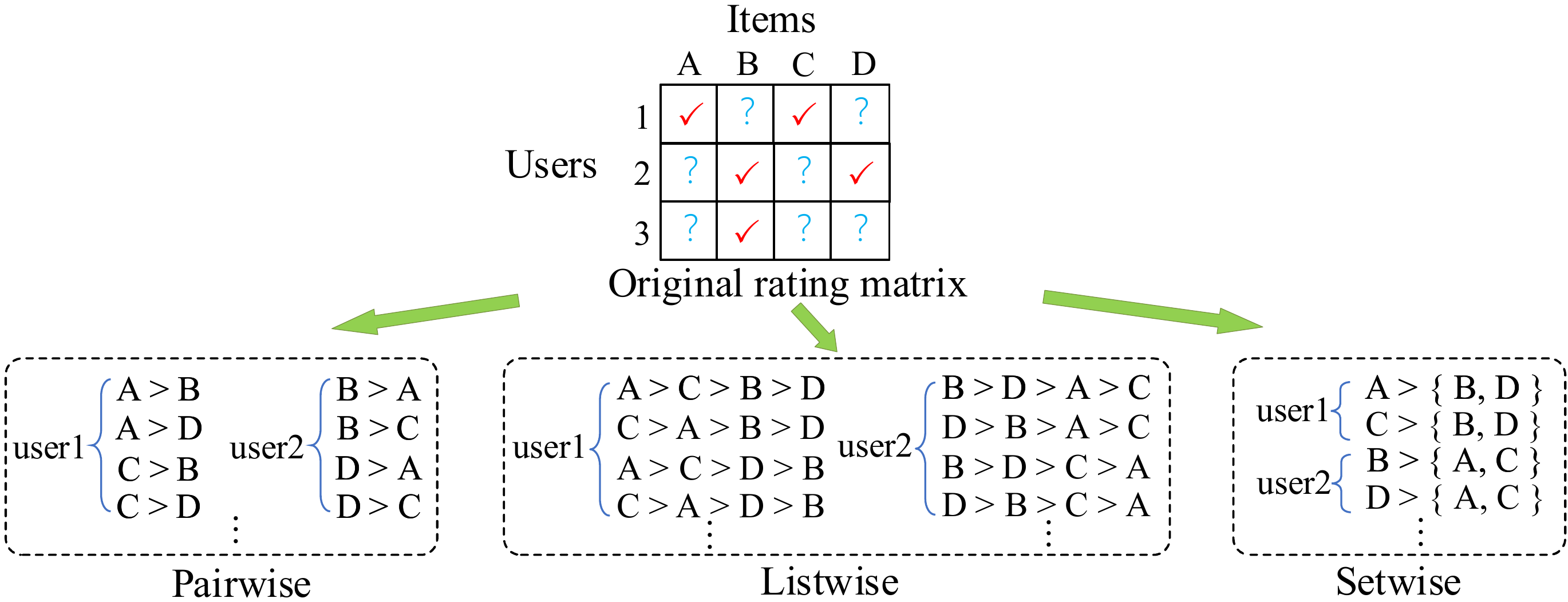}
	\caption{The diagrammatic sketch of preference structures in different collaborative ranking approaches, where the notation "\textgreater" represents the preference order.}
	\label{fig:ranking}
\end{figure*}

While considerable efforts have been made in this direction~\cite{rendle2009bpr,chen2009ranking,weimer2008cofi}, represented by the well-known \emph{pairwise} and \emph{listwise} approaches, some critical challenges still exist. As for the family of pairwise approaches~\cite{rendle2009bpr,freund2003efficient,chapelle2010efficient,wu2017large,krohn2012multi}, which take the item pair as the basic element to model the preference structure in implicit feedback, they are prone to the inconsistency problem between assumption and practice. For example, Bayesian Personalized Ranking (BPR)~\cite{rendle2009bpr}, one of the most widely used collaborative pairwise approaches, tries to maximize the probability of binary comparison between positive and unobserved feedback. Such treatment requires the strict assumption of independent pairwise preference over two items as the basis for constructing pairwise loss. However, as shown in Figure~\ref{fig:ranking}, if there exist the item preference pairs ``A\textgreater B'' and ``C\textgreater D'' for user~1, the pairs ``A\textgreater D'' and ``C\textgreater B'' must also exist for user~1 due to the binary value of implicit feedback. In other words, we have $ p(\text{A\textgreater D,C\textgreater B}|\text{A\textgreater B,C\textgreater D})=1 $ in the practical pair construction process, which breaks the independence among pairs and thus influences the optimization result of the pairwise loss. 
Some follow-up studies chose to relax the independence assumption by considering group information. For example, GBPR~\cite{pan2013gbpr} introduced richer users' interactions and Cofiset~\cite{pan2013cofiset} defined a user’s preference on the item group to consider the composition effect. However, the inconsistency problem remains to some extent.

As for the listwise approaches, the key challenge is how to efficiently accommodate ``ties'' (items with the same rating value) due to the precondition of entire list permutation, since there is no clear sequential relationship but binary rating in implicit feedback. Besides, they measure the uncertainty between the top-$ P $ items on the observed and predicted list by calculating the cross-entropy~\cite{cao2007learning,shi2010list,huang2015listwise,wang2016ranking}, which would result in the exponential computational complexity to $ P $ (that is why often $ P$ is set as $ 1 $). Though \citet{wu2018sql} tried to propose a permutation probability based listwise model to solve the above challenges, only the upper bound rather than the original negative log-likelihood is optimized.

To avoid the limitations of the existing collaborative ranking approaches, in this paper, we propose a novel setwise Bayesian approach, namely SetRank, for collaborative ranking. SetRank has the ability in accommodating the characteristics of implicit feedback in recommender systems. Particularly, we first make a weaker independence assumption compared to pairwise approaches, that is, each user prefers every positive item over the set of unobserved items independently. Hence, we can transform the original rating records into the comparisons between each single positive item and the set of unobserved items, which could avoid the inconsistency problem in pairwise approaches, as the example shown in Figure 1. Moreover, since there is no ordering information between unobserved items, it is unnecessary to rank the set of unobserved items, which relaxes the permutation form in listwise approaches. Specifically, our approach is named as ``setwise'' because the preference order of a user is only defined between each positive item and the set of unobserved items. Consequently, SetRank is able to model the properties of implicit feedback in a more effective manner, with avoiding the disadvantages of both pairwise and listwise ranking approaches. The contributions of this work can be summarized as follows:

\begin{itemize}
	\item We propose a novel setwise Bayesian collaborative ranking approach, namely SetRank, to provide a new research perspective for implicit feedback based recommendations. SetRank can inherently accommodate the characteristics of implicit feedback. 
	\item We design two implementations for SetRank, namely MF-SetRank and Deep-SetRank based on matrix factorization and neural networks, respectively.
	\item  We validate our approach by both theoretical analysis and experiments. Specifically, we prove that the bound of excess risk can be bounded by a $\sqrt{M/N}$ term, where $M$ and $N$ are the numbers of items and users, respectively. Meanwhile, extensive experiments on four real-world datasets clearly demonstrate the advantages of our approach compared with various state-of-the-art baselines.
\end{itemize}

\section{Setwise Bayesian Collaborative Ranking}

\subsection{Problem Formulation} 
Suppose there are $ N $ users and $ M $ items in the dataset. Let $ P_i $ and $ O_i $ denote the set of positive and unobserved items for each user $ i $, respectively. User $ i $ has $ J_i = |P_i| $ positive items and $ K_i = |O_i| $ unobserved items.  Then the rating matrix $ R = \{R_{il}\}_{N \times M} $ is a binary matrix, i.e., $ R_{ij} = 1 $ for $ j \in P_i $ and $ R_{ik} = 0 $ for $ k \in O_i $. The goal of collaborative ranking is to recommend each user an ordered item list by predicting the preference score matrix $ X = \{X_{il}\}_{N \times M} $. 

\subsection{SetRank Optimization Criterion} 
The target of SetRank is to maximize the posterior probability of preference structure to build the Bayesian formulation of collaborative ranking:

\begin{small}
	\begin{equation}\label{equ:1}
	p(\Theta|>_{total})\propto p(>_{total}|\Theta)p(\Theta),
	\end{equation}
\end{small}

\noindent where $ >_{total}\; = \{>_i\}_{i=1}^N $ and $ >_i $ is a random variable representing the preference structure of user $ i $, which takes values from all possible preference structures. $ \Theta $ is the model parameters to be learned. 

Before modeling the setwise preference structure probability, we first give a new independence assumption: 

\begin{assumption}\label{assum1}
	Every user $ i $ prefers the positive item $ j\in P_i $ to unobserved item set $ O_i $ independently.
\end{assumption}

In this setwise assumption, we ignore the direct comparisons among positive or unobserved items to better reflect the nature of implicit feedback, since there is no explicit item-level preference information. Supposing there is only one user, we have no reason to decide which positive item is better than another positive one, or which unobserved item is better than another unobserved one. Only when there are many users, we can then exploit collaborative information to derive entire ranking results. 

By comparison, pairwise approaches like BPR~\cite{rendle2009bpr} establish the individual binary comparison over each positive item and each unobserved item, which need the strict assumption that item comparisons are independent for optimization. However, in the pair construction process, the pairs are bound to be dependent due to the characteristic of implicit feedback. When calculating the pairwise loss, pairwise approaches still assume pairs are independent and optimize the improper loss. By contrast, setwise approach has no such inconsistency problem owing to the weaker independence assumption.

Moreover, setwise permutation form is weaker than listwise approaches. In the comparisons, we do not care about the ranking of items in $ O_i $, since the ordering information of unobserved items is naturally missing in implicit data. As a result, all the unobserved items are treated equally in the preference comparison, and thus the setwise ranking approach is inherently suitable for handling implicit data. 

According to our assumption, we can transform Equation~\ref{equ:1} into the following form:

\begin{small}
	\begin{equation}
	p(>_{total}|\Theta) = \prod_{i=1}^N p(>_i|\Theta) = \prod_{i=1}^N\prod_{j \in P_i} p(j>_i O_i|\Theta),
	\end{equation}
\end{small}

\noindent where $ j>_i O_i $ denotes the user $ i $ prefers item $ j $ to item set $ O_i $. Therefore, we turn to collect the preference comparison between a single positive item and an unobserved item set. For example, in Figure~\ref{fig:ranking}, there are two comparisons for user~1, $ A > \{B,D\} $ and $ C > \{B,D\} $. 

In the setwise preference structure, the positive item $ j $ and the unobserved set $ O_i $ compose a new item list $ L_{ij} $. Hence, it is convenient to draw the concept of permutation probability~\cite{cao2007learning} from listwise approaches for further specifying the preference structure probability $ p(j>_i O_i) $. Review that in listwise approach, a permutation $ \pi = \{\pi_1, \pi_2,..., \pi_m\} $ is a list in descending order of the $ m $ items~\cite{cao2007learning}. Denote the scores assigned to items as a vector $ s=(s_1,s_2,...,s_m) $ and $ \phi(x) $ is an increasing and strictly positive function. Then the permutation probability is defined as:

\begin{small}
	\begin{equation}\label{equ:permutation}
	p_s(\pi) := \prod_{d=1}^m \frac{\phi(s_{\pi_d})}{\sum_{l=d}^m\phi(s_{\pi_l})}.
	\end{equation}
\end{small}

It is easy to verify that $ p_s(\pi) $ is a valid probability distribution. In the literature, permutation probability has been widely used in many listwise approaches to calculate the cross entropy due to many beneficial properties~\cite{xia2008listwise}. These properties guarantee that items with higher scores are more likely to be ranked higher. However, a serious problem of the definition is that we have to calculate $ P! $ permutation probabilities to obtain the top-$ P $ probability of the list. Fortunately, in our case, we only need to place the positive item $ j $ at the top of List $ L_{ij} $, which means we just concentrate on the top-1 probability. 
Actually, \citet{cao2007learning} had proved that the top-$ 1 $ probability $ p_{s, 1}(d) $ of item $ d $ can be efficiently calculated under the definition of Equation~\ref{equ:permutation} as follows:  

\begin{small}
	\begin{equation}\label{equ:topone}
	p_{s, 1}(d) =  \frac{\phi(s_{d})}{\sum_{l=1}^m\phi(s_l)}.
	\end{equation}
\end{small}

With the help of Equation~\ref{equ:topone} and the preference score matrix $ X $, now we can give the detailed formulation of the setwise preference probability over all users:

\begin{small}
	\begin{equation}\label{equ:semi}
	p(>_{total}|\Theta) = \prod_{i=1}^N \prod_{j\in P_i}\frac{\phi(X_{ij})}{\phi(X_{ij})+ \sum_{k\in O_i}\phi(X_{ik})}.
	\end{equation}
\end{small}

As one can see, Equation~\ref{equ:semi} indicates that if positive items have higher scores and unobserved items have lower scores, this preference structure will be more likely to be true. 

At last, to complete the Bayesian inference, we introduce a general prior probability $ p(\Theta) $. Following BPR~\cite{rendle2009bpr}, $ p(\Theta) $ is set as a normal distribution with zero mean and variance-covariance matrix $ \lambda_\Theta I $. Hence, maximizing the posterior probability is equivalent to minimizing the following function:

\begin{small}
	\begin{equation}\label{equ:obj}
	L = \sum_{i=1}^{N}\sum_{j\in P_i} -\log p(j>_i O_i|\Theta) + \lambda_\Theta \|\Theta\|^2.
	\end{equation}
\end{small}

Note that though some listwise approaches also exploit Equation~\ref{equ:topone}~\cite{cao2007learning,shi2010list}, it is actually quite different from SetRank. First, listwise approaches are essentially based on the top-$ P $ probability since they consider the order in a list composed of multiple positive and unobserved items. In fact, using a larger $ P $ tends to improve the performance of listwise approaches~\cite{cao2007learning}. They use the top-1 probability mainly due to the compromise of exponential computational complexity for calculating the top-$ P $ probability. However, our setwise assumption, which is more appropriate for implicit feedback, is naturally based on the top-1 probability. Second, they could only employ cross-entropy loss for calculation while the cross-entropy loss may rank worse scoring permutations higher~\cite{wu2018sql}. On the contrary, our loss is strictly obtained by Bayesian inference without the adoption of cross-entropy.

\subsection{Implementation} \label{implement}
It is quite flexible to apply many well-known models to learn the score matrix X. In the literature, matrix factorization~\cite{mnih2008probabilistic} and neural network (NN)~\cite{xue2017deep} have demonstrated their effectiveness and practicability for recommender systems. 
Therefore, in this paper, we introduce two implementations for SetRank, namely MF-SetRank and Deep-SetRank, based on the above two models, respectively.

\noindent \textbf{MF-SetRank.} MF-SetRank is based on the popular collaborative model, Probabilistic Matrix Factorization (PMF)~\cite{mnih2008probabilistic}. PMF factorizes the score matrix into two factor matrices representing user and item latent features. Along this line, we have $ X = U^TV $, where $ U \in \mathbb{R}^{r \times N}$ and $ V \in \mathbb{R}^{r \times M}$ are latent user and item matrices, respectively. Then the prior probabilities over columns of $ U , V $ are assumed to be the normal distribution, i.e., $ p(u_i)\sim \mathcal{N}(0,\lambda^{-1}I) $ and $ p(v_l)\sim \mathcal{N}(0,\lambda^{-1}I) $, where $ \lambda $ is the regularization parameter. In this way, we can transform Equation~\ref{equ:obj} to the following form:

\begin{small}
	\begin{align}\label{equ:object}
	\nonumber L = \sum_{i=1}^{N}&\sum_{j\in P_i} -\log \frac{\phi(u_i^Tv_j)}{\phi(u_i^Tv_j)+ \sum_{k\in O_i}\phi(u_i^Tv_k)}\\
	&+ \frac{\lambda}{2}(\sum_{i=1}^{N}\|u_i\|^2 + \sum_{l=1}^{M}\|v_l\|^2).
	\end{align}
\end{small}

For the ease of calculation, we let $ \log \phi(x)  $ be the sigmoid function, i.e., $ \log \phi(x) = \sigma(x) = 1/(1 +e^{-x})$. It is easy to verify that such $ \phi(x) $ is an increasing and strictly positive function. Besides, this choice is also beneficial for bounding the excess risk which we will discuss in the next subsection. 

\begin{algorithm}[!t] 
	\small
	\caption{Gradient update for V when fixing U}
	\label{alg:Updating}
	\begin{algorithmic}[1]
		\REQUIRE $ V, U, \gamma, decay, \lambda, \{P_i, \tilde{O}_i| 1\leq i\leq N\} $
		\ENSURE $ V $
		\STATE $ grad = \lambda \cdot V $
		\FOR{$i=1$ to $N$} 
		\STATE Precompute $ g_l = u_i^Tv_l $, for $  \forall l\in P_i \cup \tilde{O}_i $
		\STATE Initialize $ tmp = 0, totalsum = 0, sum = 0, s[l] = 0$ for $\forall l\in P_i $, $c[l] = 0$ for $\forall l\in P_i \cup \tilde{O}_i $
		\FOR{$ l \in \tilde{O}_i $}
		\STATE $ sum \ +\!\!= \exp(g_l) $  
		\ENDFOR
		\FOR{$ l \in P_i $}
		\STATE $ c[l] \ -\!\!= g_l \cdot (1-g_l) $ 
		\STATE $ s[l] = sum + \exp(g_l)$ 
		\STATE $ totalsum \ +\!\!= 1/s[l]$
		\ENDFOR
		\FOR{$ l \in \tilde{O}_i $}
		\STATE $ c[l] \ +\!\!= \exp(g_l) \cdot g_l \cdot (1-g_l) \cdot totalsum $  
		\ENDFOR
		\FOR{$ l \in P_i $}
		\STATE $ c[l] \ +\!\!= \exp(g_l) \cdot g_l \cdot (1-g_l) /s[l] $  
		\ENDFOR
		\FOR{$ l \in  P_i \cup \tilde{O}_i $}
		\STATE $ grad[:, l] \ +\!\!= c[l] \cdot u_i$
		\ENDFOR
		\ENDFOR 
		\STATE $ V \ -\!\!= \gamma \cdot grad $
		\STATE $ \gamma \ *\!\!= decay $
		\STATE Return $ V $
	\end{algorithmic}
\end{algorithm}

Another notable thing is that we do not have to go through all the unobserved items for every user in each epoch, considering that the positive feedback is much more influential than the unobserved feedback. Following \citet{wu2018sql}, we can randomly sample $ \tilde{K}_i = \tau \cdot J_i $ unobserved items in each epoch to compose the set $ \tilde{O}_i $ for replacing $ O_i $ in Equation~\ref{equ:object}. 

\begin{figure}[t]
	\centering
	\includegraphics[width=0.4\textwidth]{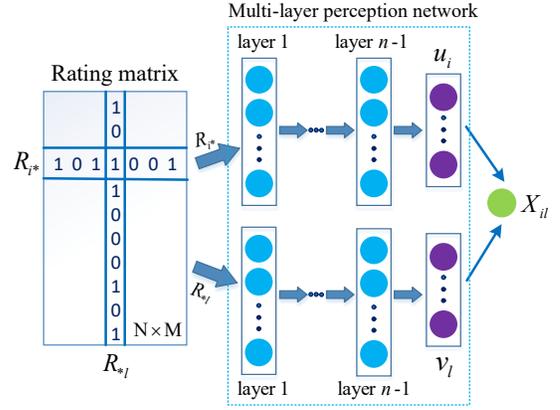}
	\caption{The modeling process of Deep-SetRank.}
	\label{fig:nn}
\end{figure}

In each epoch, we update the latent factors $ U $ and $ V $ by the gradients $ \nabla_U L $ and $ \nabla_V L $, respectively. The speed efficiency of recommender system is quite important~\cite{wu2017large}. Though a direct way to calculate the gradients costs $ O(NJ\tilde{K}r)=O(NJ^2r) $ time, where $ J=\max \{J_i,1\leq i \leq N\} $ and $ \tilde{K}=\tau\cdot J $, there is actually numerous repeated calculations here. We provide a cleverer approach in Algorithm~\ref{alg:Updating} to rearrange the computation so that it only requires $O(N(J+\tilde{K})r)=O(NJr) $ time. Let us take the process of updating $ V $ as an example. Specifically, given a fixed latent user matrix $ U $, the regularization parameter $ \lambda $, the positive item set $ P_i $, the unobserved item set $ \tilde{O}_i $ and the decaying rate $ decay $ of the step size $ \gamma $, Algorithm~\ref{alg:Updating} shows how to update the gradients for latent item matrix $ V $. Thus, MF-SetRank could run with a linear computational complexity, which is same as the efficient rating prediction methods based on matrix factorization~\cite{mnih2008probabilistic,hu2008collaborative}. 

\noindent \textbf{Deep-SetRank.} In recent years, neural networks have shown good capacity on non-linear projection and embedding in recommender systems~\cite{qin2019duerquiz,he2017neural}. Inspired by Deep Matrix Factorization (DeepMF)~\cite{xue2017deep}, we design a NN based setwise model called Deep-SetRank. 

As shown in Figure~\ref{fig:nn}, Deep-SetRank transforms the row $ R_{i*} $ and column $ R_{*l} $ of rating matrix $ R $ to obtain latent user and item matrices by user and item embedding networks, respectively. Then we still employ Equation~\ref{equ:object} as the setwise loss function. Following DeepMF, we choose the multi-layer perception network (MLP) as the embedding network. Take user network as an example, we have

\begin{small}
	\begin{align}
	\nonumber h_1 &= f_1(W_1R_{i*}+b_1),&\\
	\nonumber h_t &= f_t(W_th_{t-1}+b_t), \ \quad t\in [2,n-1],&\\
	u_i &= f_{n}(W_{n}h_{n-1}+b_{n}),&
	\end{align}
\end{small}

\noindent where $ h_t $ is the $ t$-th hidden layer with weight matrix $ W_t $ and bias term $ b_t $. For the activation function $ f_t(\cdot) $, we employ the \emph{sigmoid} function for the first $ n-1 $ layers and the \emph{tanh} function for the last layer.
Hence, we can predict the scores by the product of these two matrices. For each user, we simultaneously calculate the scores of items in both $ P_i $ and $ \tilde{O}_i $ in a batch for optimizing setwise loss. Different from MF-SetRank, Deep-SetRank needs to train two neural networks rather than latent matrices.

\subsection{Theoretical Analysis} \label{section:theory}
In this subsection, we aim at giving the theoretical bound for the excess risk, i.e., the expected difference between the estimate and the truth, of SetRank. Without loss of generality, we assume that all the users have the same number of positive items and unobserved items for the sake of convenience. Hence, we have $ J = J_i, K = K_i$ for $\forall i $. Note that the result can be readily generalized to the individual setting. 

Considering the following constrained optimization of a general setwise method:

\begin{small}
	\begin{equation}\label{equ:opt}
	\hat{X} := \arg\min_X -\log p(>_{total}|X) \ \text{such that}\ X \in \mathcal{X},
	\end{equation}
\end{small}

\noindent where $ \mathcal{X} $ is the feasible set. Usually, X is constrained by the norm regularization to satisfy the low-rank condition. For example, in the personalized collaborative setting, $ \mathcal{X} = \{X|X=U^TV, \|U\|_F\leq c_u, \|V\|_F\leq c_v\} $. Here $ \|\cdot\|_F $ is the Frobenius norm. Supposing there is a $ X^*\in \mathcal{X}$ such that $ >_{total} $ is generated from $ p(>_{total}|X^*)$. Then the excess risk is given in the form of KL divergence between the real and estimated probability: $ 	D(X^*,\hat{X}):=\frac{1}{N}\sum_{i=1}^N \mathbb{E} \log \frac{p(>_i|X^*_i)}{p(>_i|\hat{X}_i)}. $

So far, the state-of-the-art listwise method could bound the excess risk by $ O_\mathbb{P}\left(\sqrt{rM/N}\ln M\right) $ in the personalized collaborative setting~\cite{wu2018sql}. Here we will show that the bound of SetRank is $ O_\mathbb{P}\left(\sqrt{rM/N}(1+J/K)\right) $ owing to the weaker precondition. In practice, the positive feedback always accounts for merely a tiny fraction of the total items. So, we have $ J/K \ll 1 $, which makes the result sound. 

First, we give another statistical interpretation of Equation~\ref{equ:semi} from the generative perspective:

\begin{thm}\label{1}
	Suppose there is a matrix $ Y= \{Y_{il}\}_{N \times M}$. Each entry $ Y_{il} $ is independently drawn from an exponential distribution with rate $ \phi(X_{il}) $. For each row $ Y_i $, let the $ J $ smallest entries form the set $ P_i $ and others form the set $ O_i $. Then the ranking structure probability $ p(>_{total}|X) $, i.e., the probability that entries in $ P_i $ are less than those in $ O_i $ , is exactly equal to the RHS of Equation~\ref{equ:semi}.
\end{thm}

The proof for Theorem~\ref{1} can be found in the Appendix. From Theorem~\ref{1}, we know that the setwise preference probability can also be seen as the probability of a ranking structure over the matrix $ Y $ composed of $ N \times M $ independent exponential random variables. Thus, we could give the following theorem according to McDiarmid's inequality~\cite{mcdiarmid1989method} and Dudley’s chaining~\cite{talagrand2006generic}:

\begin{thm}\label{3}
	Let $ \mathcal{Z} :=\{\log \phi(X)|X\in \mathcal{X}\} $ be the image of element-wise function $ \log \phi(X) $ and $ \|\cdot\|_{\infty,2} $ be the $ \infty,2 $ norm defined as $ \|Z\|_{\infty,2} := \sqrt{\sum_{i=1}^N \|Z_i\|^2_\infty}, Z \in \mathbb{R}^{N \times M}$. Denote $ \mathcal{N}(\epsilon, \mathcal{Z}, \|\cdot\|_{\infty,2}) $ as the $ \epsilon $-covering number of $ \mathcal{Z} $ in $ \infty,2 $ norm, which represents the fewest number of spherical balls of radius $ \epsilon $ needed to completely cover $ \mathcal{Z} $ in the condition of $ \infty,2 $ norm. Hence, if $ \|Z_{i}\|_{\infty} \leq \alpha$ for $\ \forall\ i$, we have
	
	\begin{small}
		\begin{equation}
		D(X^*,\hat{X}) = O_\mathbb{P}\left(g(\mathcal{Z})\sqrt{M}/N\left(1+J/K\right) \right),
		\end{equation}
	\end{small}
	
	\noindent where $ g(\mathcal{Z}) = \int_{0}^{\infty} \sqrt{\ln \mathcal{N}(u, \mathcal{Z}, \|\cdot\|_{\infty,2})} du $.
\end{thm}

By Theorem~\ref{3}, we are able to obtain a bound in the general setting of $ X $. Particularly, in the personalized collaborative setting, we can obtain the further result as follows:

\begin{thm}\label{4}
	Suppose that $ \|Z_{i}\|_{\infty} \leq \alpha$ for $\ \forall\ i$ and $ \log\phi(x) $ is 1-Lipschitz, then in the personalized collaborative setting, we have
	
	\begin{small}
		\begin{equation}
		D(X^*,\hat{X}) = O_\mathbb{P}\left(\sqrt{rM/N}\left(1+J/K\right)\right).
		\end{equation}
	\end{small}
	
\end{thm}

The detailed proofs for Theorem~\ref{3} and~\ref{4} are in the Appendix. Theorem~\ref{4} shows that when fixing the rank of latent factors, we will have a better estimate with a larger number of users and a smaller number of items, which is accord with the intuition. Besides, a smaller $ J $ with a larger $ K $ is beneficial for bounding the excess risk.

\section{Experiments}

\begin{table*}[bt]
	\caption{The recommendation performance of different approaches. (Methods with notation $ ^* $ are our proposed methods. We conducted the paired t-tests to verify that all improvements by SetRank are statistically significant for $ p < 0.05 $.)}
	\label{tab1}
	\centering
	\resizebox{!}{7cm}{
		\begin{tabular}{llcccccc}
			\toprule
			Datasets & Methods & P@5 & P@10 & R@5 & R@10 & MAP@5 & MAP@10 \\ 
			\midrule
			\multirow{8}{1.4cm}{\emph{MovieLens}} & WMF  & $ {0.5161}_{\pm 0.0050} $ & $ {0.4787}_{\pm 0.0039} $ & $ {0.0173}_{\pm 0.0002} $ & $ {0.0317}_{\pm 0.0003} $ & $ {0.4125}_{\pm 0.0066} $ & $ {0.3449}_{\pm 0.0051} $       \\ 
			&BPR	& $ {0.6330}_{\pm 0.0075} $ & $ {0.6058}_{\pm 0.0075} $ & $ {0.0213}_{\pm 0.0004} $ & $ {0.0401}_{\pm 0.0007} $ & $ {0.5403}_{\pm 0.0105} $ & $ {0.4835}_{\pm 0.0078} $ 	\\  
			&Cofiset	& $ {0.6500}_{\pm 0.0050} $ & $ {0.6152}_{\pm 0.0033} $ & $ {0.0217}_{\pm 0.0003} $ & $ {0.0405}_{\pm 0.0002} $ & $ {0.5657}_{\pm 0.0060} $ & $ {0.5020}_{\pm 0.0042} $ 	\\ 
			&SQL-Rank	& $ {0.6609}_{\pm 0.0020} $ & $ {0.6227}_{\pm 0.0028} $ & $ {0.0227}_{\pm 0.0003} $ & $ {0.0421}_{\pm 0.0003} $ & $ {0.5749}_{\pm 0.0018} $ & $ {0.5081}_{\pm 0.0025} $ 	\\ 
			&MF-SetRank$ ^* $ & $ {0.6762}_{\pm 0.0024} $ & $ {0.6398}_{\pm 0.0017} $ & $ {0.0231}_{\pm 0.0002} $ & $ {0.0427}_{\pm 0.0001} $ & $ {0.5940}_{\pm 0.0028} $ & $ {0.5296}_{\pm 0.0033} $ 	\\
			&DeepMF  & $ {0.6191}_{\pm 0.0136} $ & $ {0.5831}_{\pm 0.0120} $ & $ {0.0208}_{\pm 0.0006} $ & $ {0.0383}_{\pm 0.0009} $ & $ {0.5270}_{\pm 0.0159} $ & $ {0.4627}_{\pm 0.0143} $       \\
			&Deep-BPR	& $ {0.6624}_{\pm 0.0061} $ & $ {0.6242}_{\pm 0.0067} $ & $ {0.0225}_{\pm 0.0003} $ & $ {0.0416}_{\pm 0.0008} $ & $ {0.5798}_{\pm 0.0098} $ & $ {0.5141}_{\pm 0.0077} $ 	\\ 
			&Deep-SQL	& $ {0.6794}_{\pm 0.0074} $ & $ {0.6432}_{\pm 0.0024} $ & $ {0.0233}_{\pm 0.0004} $ & $ {0.0431}_{\pm 0.0003} $ & $ {0.5990}_{\pm 0.0076} $ & $ {0.5343}_{\pm 0.0036} $ 	\\ 
			&Multi-VAE  & $ {0.6806}_{\pm 0.0012} $ & $ {0.6429}_{\pm 0.0030} $ & $ {0.0235}_{\pm 0.0002} $ & $ {0.0434}_{\pm 0.0002} $ & $ {0.5996}_{\pm 0.0017} $ & $ {0.5338}_{\pm 0.0020} $       \\ 
			&Deep-SetRank$ ^* $	& $ \textbf{0.6956}_{\pm 0.0030} $ & $ \textbf{0.6557}_{\pm 0.0013} $ & $ \textbf{0.0242}_{\pm 0.0003} $ & $ \textbf{0.0447}_{\pm 0.0003} $ & $ \textbf{0.6150}_{\pm 0.0039} $ & $ \textbf{0.5473}_{\pm 0.0019} $	\\
			\midrule
			\multirow{8}{1.4cm}{\emph{Kindle}}     & WMF  & $ {0.1281}_{\pm 0.0018} $ & $ {0.1105}_{\pm 0.0015} $ & $ {0.0577}_{\pm 0.0014} $ & $ {0.0957}_{\pm 0.0015} $ & $ {0.0830}_{\pm 0.0008} $ & $ {0.0580}_{\pm 0.0007} $       \\ 
			&BPR	& $ {0.1372}_{\pm 0.0030} $ & $ {0.1174}_{\pm 0.0024} $ & $ {0.0660}_{\pm 0.0023} $ & $ {0.1078}_{\pm 0.0030} $ & $ {0.0896}_{\pm 0.0019} $ & $ {0.0623}_{\pm 0.0012} $ 	\\
			&Cofiset	& $ {0.1438}_{\pm 0.0025} $ & $ {0.1223}_{\pm 0.0017} $ & $ {0.0690}_{\pm 0.0017} $ & $ {0.1127}_{\pm 0.0020} $ & $ {0.0940}_{\pm 0.0017} $ & $ {0.0653}_{\pm 0.0011} $ 	\\  
			&SQL-Rank	& $ {0.1478}_{\pm 0.0024} $ & $ {0.1238}_{\pm 0.0008} $ & $ {0.0705}_{\pm 0.0010} $ & $ {0.1135}_{\pm 0.0009} $ & $ {0.0986}_{\pm 0.0017} $ & $ {0.0678}_{\pm 0.0009} $ 	\\ 
			&MF-SetRank$ ^* $	& $ {0.1773}_{\pm 0.0027} $ & $ {0.1466}_{\pm 0.0019} $ & $ {0.0869}_{\pm 0.0010} $ & $ {0.1375}_{\pm 0.0016} $ & $ {0.1210}_{\pm 0.0025} $ & $ {0.0834}_{\pm 0.0015} $	\\
			&DeepMF  & $ {0.1484}_{\pm 0.0022} $ & $ {0.1236}_{\pm 0.0011} $ & $ {0.0679}_{\pm 0.0013} $ & $ {0.1085}_{\pm 0.0016} $ & $ {0.0992}_{\pm 0.0017} $ & $ {0.0684}_{\pm 0.0010} $       \\ 
			&Deep-BPR	& $ {0.1654}_{\pm 0.0025} $ & $ {0.1372}_{\pm 0.0013} $ & $ {0.0824}_{\pm 0.0010} $ & $ {0.1300}_{\pm 0.0017} $ & $ {0.1106}_{\pm 0.0024} $ & $ {0.0765}_{\pm 0.0014} $ 	\\ 
			&Deep-SQL	& $ {0.1754}_{\pm 0.0045} $ & $ {0.1435}_{\pm 0.0023} $ & $ {0.0842}_{\pm 0.0030} $ & $ {0.1320}_{\pm 0.0036} $ & $ {0.1201}_{\pm 0.0032} $ & $ {0.0823}_{\pm 0.0017} $ 	\\ 
			&Multi-VAE  & $ {0.1651}_{\pm 0.0019} $ & $ {0.1360}_{\pm 0.0010} $ & $ {0.0800}_{\pm 0.0013} $ & $ {0.1265}_{\pm 0.0011} $ & $ {0.1120}_{\pm 0.0015} $ & $ {0.0771}_{\pm 0.0009} $       \\ 
			&Deep-SetRank$ ^* $	& $ \textbf{0.1837}_{\pm 0.0020} $ & $ \textbf{0.1502}_{\pm 0.007} $ & $ \textbf{0.0894}_{\pm 0.0009} $ & $ \textbf{0.1402}_{\pm 0.0013} $ & $ \textbf{0.1266}_{\pm 0.0015} $ & $ \textbf{0.0867}_{\pm 0.0009} $	\\
			\midrule
			\multirow{8}{1.4cm}{\emph{Yahoo}}    & WMF  & $ {0.1813}_{\pm 0.0021} $ & $ {0.1465}_{\pm 0.0011} $ & $ {0.1297}_{\pm 0.0015} $ & $ {0.2048}_{\pm 0.0017} $ & $ {0.1182}_{\pm 0.0015} $ & $ {0.0772}_{\pm 0.0010} $       \\ 
			&BPR	& $ {0.2096}_{\pm 0.0032} $ & $ {0.1729}_{\pm 0.0020} $ & $ {0.1475}_{\pm 0.0031} $ & $ {0.2389}_{\pm 0.0039} $ & $ {0.1378}_{\pm 0.0029} $ & $ {0.0928}_{\pm 0.0010} $ 	\\
			&Cofiset	& $ {0.2196}_{\pm 0.0041} $ & $ {0.1791}_{\pm 0.0018} $ & $ {0.1554}_{\pm 0.0034} $ & $ {0.2478}_{\pm 0.0032} $ & $ {0.1457}_{\pm 0.0034} $ & $ {0.0979}_{\pm 0.0016} $ 	\\  
			&SQL-Rank	& $ {0.2137}_{\pm 0.0031} $ & $ {0.1723}_{\pm 0.0011} $ & $ {0.1502}_{\pm 0.0025} $ & $ {0.2380}_{\pm 0.0017} $ & $ {0.1432}_{\pm 0.0026} $ & $ {0.0951}_{\pm 0.0012} $ 	\\
			&MF-SetRank$ ^* $	& $ {0.2267}_{\pm 0.0012} $ & $ {0.1817}_{\pm 0.0007} $ & $ {0.1616}_{\pm 0.0008} $ & $ {0.2540}_{\pm 0.0018} $ & $ {0.1528}_{\pm 0.0016} $ & $ {0.1009}_{\pm 0.0006} $	\\
			&DeepMF  & $ {0.2167}_{\pm 0.0019} $ & $ {0.1764}_{\pm 0.0009} $ & $ {0.1529}_{\pm 0.0021} $ & $ {0.2437}_{\pm 0.0032} $ & $ {0.1445}_{\pm 0.0024} $ & $ {0.0964}_{\pm 0.0012} $       \\ 
			&Deep-BPR	& $ {0.2260}_{\pm 0.0019} $ & $ {0.1831}_{\pm 0.0010} $ & $ {0.1615}_{\pm 0.0008} $ & $ {0.2549}_{\pm 0.0020} $ & $ {0.1513}_{\pm 0.0018} $ & $ {0.1013}_{\pm 0.0009} $ 	\\ 
			&Deep-SQL	& $ {0.2278}_{\pm 0.0023} $ & $ {0.1817}_{\pm 0.0013} $ & $ {0.1613}_{\pm 0.0025} $ & $ {0.2517}_{\pm 0.0029} $ & $ {0.1542}_{\pm 0.0018} $ & $ {0.1021}_{\pm 0.0009} $ 	\\ 
			&Multi-VAE  & $ {0.2300}_{\pm 0.0018} $ & $ {0.1855}_{\pm 0.0016} $ & $ {0.1631}_{\pm 0.0011} $ & $ {0.2572}_{\pm 0.0014} $ & $ {0.1550}_{\pm 0.0025} $ & $ {0.1036}_{\pm 0.0020} $       \\ 
			&Deep-SetRank$ ^* $	& $ \textbf{0.2372}_{\pm 0.0011} $ & $ \textbf{0.1894}_{\pm 0.0015} $ & $ \textbf{0.1694}_{\pm 0.0007} $ & $ \textbf{0.2637}_{\pm 0.0022} $ & $ \textbf{0.1608}_{\pm 0.0016} $ & $ \textbf{0.1067}_{\pm 0.0013} $	\\
			\midrule
			\multirow{8}{1.4cm}{\emph{CiteULike}} & WMF  & $ {0.1714}_{\pm 0.0018} $ & $ {0.1447}_{\pm 0.0010} $ & $ {0.0539}_{\pm 0.0009} $ & $ {0.0866}_{\pm 0.0009} $ & $ {0.1195}_{\pm 0.0023} $ & $ {0.0859}_{\pm 0.0012} $       \\  
			&BPR	& $ {0.1876}_{\pm 0.0027} $ & $ {0.1612}_{\pm 0.0016} $ & $ {0.0644}_{\pm 0.0009} $ & $ {0.1056}_{\pm 0.0012} $ & $ {0.1320}_{\pm 0.0027} $ & $ {0.0964}_{\pm 0.0019} $ 	\\
			&Cofiset	& $ {0.1881}_{\pm 0.0022} $ & $ {0.1590}_{\pm 0.0009} $ & $ {0.0664}_{\pm 0.0013} $ & $ {0.1058}_{\pm 0.0016} $ & $ {0.1317}_{\pm 0.0018} $ & $ {0.0953}_{\pm 0.0008} $ 	\\
			&SQL-Rank	& $ {0.1801}_{\pm 0.0024} $ & $ {0.1522}_{\pm 0.0015} $ & $ {0.0576}_{\pm 0.0008} $ & $ {0.0932}_{\pm 0.0010} $ & $ {0.1283}_{\pm 0.0019} $ & $ {0.0927}_{\pm 0.0013} $ 	\\ 
			&MF-SetRank$ ^* $	& $ {0.2124}_{\pm 0.0016} $ & $ {0.1813}_{\pm 0.0014} $ & $ {0.0764}_{\pm 0.0007} $ & $ {0.1238}_{\pm 0.0015} $ & $ {0.1523}_{\pm 0.0018} $ & $ {0.1117}_{\pm 0.0010} $	\\
			&DeepMF  & $ {0.1853}_{\pm 0.0018} $ & $ {0.1555}_{\pm 0.0015} $ & $ {0.0600}_{\pm 0.0010} $ & $ {0.0962}_{\pm 0.0015} $ & $ {0.1296}_{\pm 0.0019} $ & $ {0.0935}_{\pm 0.0014} $       \\  
			&Deep-BPR	& $ {0.2086}_{\pm 0.0009} $ & $ {0.1754}_{\pm 0.0015} $ & $ {0.0767}_{\pm 0.0012} $ & $ {0.1224}_{\pm 0.0014} $ & $ {0.1483}_{\pm 0.0015} $ & $ {0.1075}_{\pm 0.0004} $ 	\\ 
			&Deep-SQL	& $ {0.2077}_{\pm 0.0027} $ & $ {0.1733}_{\pm 0.0016} $ & $ {0.0730}_{\pm 0.0021} $ & $ {0.1152}_{\pm 0.0015} $ & $ {0.1478}_{\pm 0.0029} $ & $ {0.1064}_{\pm 0.0017} $ 	\\
			&Multi-VAE  & $ {0.2081}_{\pm 0.0027} $ & $ {0.1720}_{\pm 0.0016} $ & $ {0.0775}_{\pm 0.0014} $ & $ {0.1205}_{\pm 0.0022} $ & $ {0.1502}_{\pm 0.0031} $ & $ {0.1068}_{\pm 0.0015} $       \\
			&Deep-SetRank$ ^* $	& $ \textbf{0.2233}_{\pm 0.0026} $ & $ \textbf{0.1856}_{\pm 0.0017} $ & $ \textbf{0.0830}_{\pm 0.0019} $ & $ \textbf{0.1300}_{\pm 0.0016} $ & $ \textbf{0.1606}_{\pm 0.0022} $ & $ \textbf{0.1159}_{\pm 0.0014} $	\\
			\bottomrule
	\end{tabular}}
\end{table*}

\subsection{Experimental Settings}

\noindent \textbf{Datasets.}  We evaluated the performance of our SetRank method on four real-world datasets, i.e., \emph{MovieLens}~\footnote{https://grouplens.org/datasets/movielens/},  \emph{Kindle}~\footnote{http://jmcauley.ucsd.edu/data/amazon/}, \emph{Yahoo}~\footnote{https://webscope.sandbox.yahoo.com/catalog.php?datatype=r} and \emph{CiteULike}~\footnote{http://www.citeulike.org}. \emph{MovieLens} is a commonly used movie recommendation dataset. \emph{Kindle} contains Amazon product ratings collected from Kindle Store. \emph{Yahoo}~\cite{marlin2009collaborative} contains ratings for songs from Yahoo! Music. \emph{CiteULike} is composed of users' collections of articles on CiteULike website. Following \citet{wu2018sql}, we took two steps for data preprocessing. First, the original data of \emph{MovieLens}, \emph{Kindle} and \emph{Yahoo} are in the form of 5-star ratings. We transformed them into implicit data, where each entry was marked as $ 1/0 $, depending on whether the ratings are greater than 3. Second, in order to make sure we have adequate positive feedback for better evaluating the recommendation algorithms, we filtered out users with less than 60, 20, 10, 10 positive items in \emph{MovieLens}, \emph{Kindle}, \emph{Yahoo} and \emph{CiteULike}, respectively. After data filtering, there are totally 3,937 users and 3,533 items with 923,473 positive entries in \emph{MovieLens}, 4,379 users and 3,774 items with 102,545 positive entries in \emph{Kindle}, 4,664 users and 921 items with 82,384 positive entries in \emph{Yahoo}, 4,123 users and 7,849 items with 135,365 positive entries in \emph{CiteULike}. 

\noindent \textbf{Evaluation protocols.} We randomly sampled $ 50\% $ of positive items for each user to construct the training set in each dataset, while the maximum number of item samples for each user was set as 10. Then, we sampled 1 positive item of each user as the validation set. Meanwhile, the rest data were used for test. In this way, we randomly split each dataset five times and reported all the results by mean values. To evaluate the performance, we adopted three widely used evaluation metrics, i.e., P@$ P $, R@$ P $ and MAP@$ P $~\cite{wu2018sql,wang2015collaborative}. For each user, P~(Precision)~@$ P $ measures the ratio of correct prediction results among top-$ P $ items to $ P $ and R~(Recall)~@$ P $ measures the ratio of correct prediction results among top-$ P $ items to all positive items. Furthermore, MAP (Mean Average Precision) @$ P $ considers the ranking of correct prediction results among top-$ P $ items. The final results of three metrics are given in the average of all users. 

\begin{figure*}[t]
	\centering
	\includegraphics[width=0.8\textwidth]{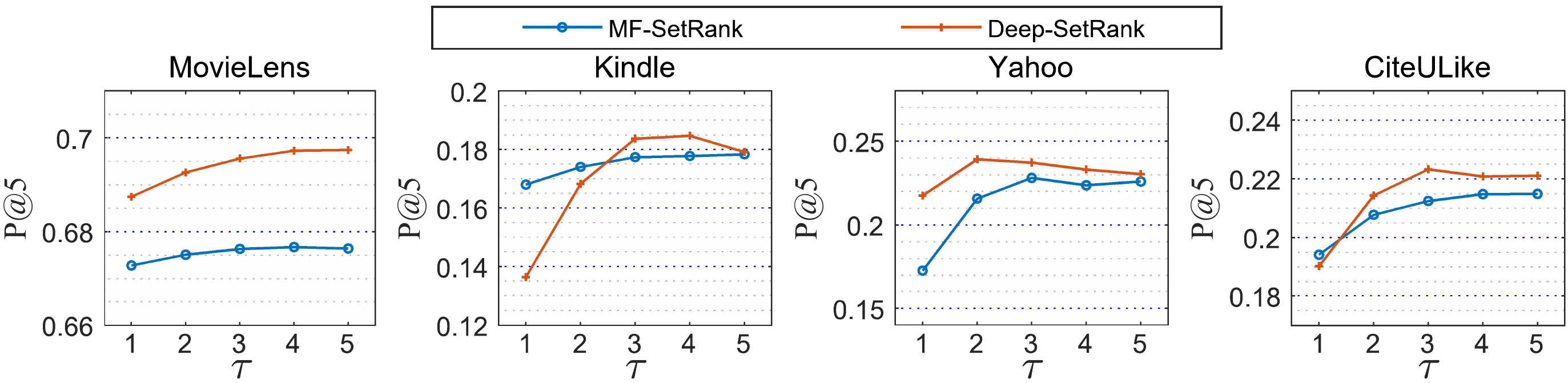}
	\caption{The performance of P@$ 5 $ with different values of sampling ratio $ \tau $ on the four datasets. }
	\label{fig:tkind}
\end{figure*}

\begin{figure*}[t]
	\centering
	\includegraphics[width=0.8\textwidth]{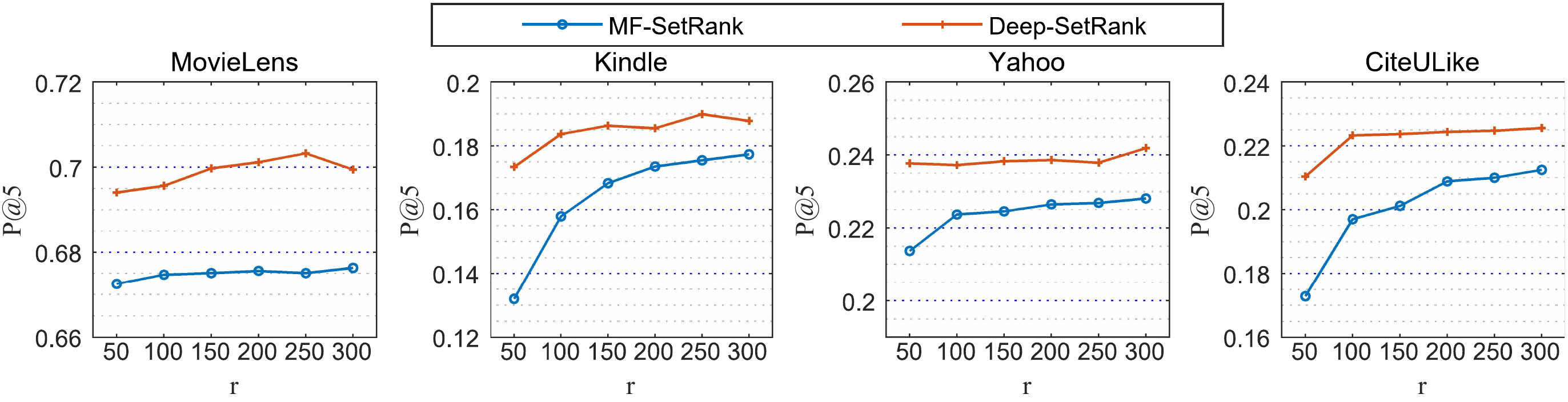}
	\caption{The performance of P@$ 5 $ with different values of dimension $ r $ on the four datasets.}
	\label{fig:rkind}
\end{figure*}

\noindent \textbf{Baselines.} The recommendation methods for comparison are listed as follows:

\begin{itemize}
	\item \textbf{WMF:} Weighted Matrix Factorization~\cite{hu2008collaborative} is a popular rating prediction method for implicit data, which introduces the confidence levels into standard matrix factorization model.
	\item \textbf{BPR:} Bayesian Personalized Ranking~\cite{rendle2009bpr} is a widely used pairwise collaborative ranking approach, which transforms the original rating matrix into the form of independent pairs.
	\item \textbf{Cofiset:} Cofiset~\cite{pan2013cofiset} defines the group preference as the mean value of each item in the group. Then the BPR loss function is used for optimization.
	\item \textbf{SQL-Rank:} Stochastic Queuing Listwise Ranking~\cite{wu2018sql} is a state-of-the-art listwise approach, which breaks ties randomly and generates multiple possible permutations.
	\item \textbf{DeepMF:} Deep Matrix Factorization~\cite{xue2017deep} is a NN based matrix factorization model.
	\item \textbf{Multi-VAE:} Variational Autoencoders for Collaborative Filtering~\cite{liang2018variational} is a state-of-the-art NN based method, which extends variational autoencoders to recommendations for implicit feedback.
	\item \textbf{Deep-BPR, Deep-SQL:} These two methods have the same network architecture with Deep-SetRank, but we replace the loss function in SetRank by those in BPR and SQL-Rank, respectively. Hence, we obtain these two NN based pairwise and listwise approaches. 
	\item \textbf{MF-SetRank, Deep-SetRank:} These two methods are our proposed setwise Bayesian approaches for collaborative ranking from implicit feedback. We release our code at \url{https://github.com/chadwang2012/SetRank}.
	
\end{itemize}

Please note that WMF, BPR, SQL-Rank, MF-SetRank are all implemented with a basic matrix factorization model and the four ``Deep'' methods are all implemented with the same neural network architecture, so it is a fair setting to compare the performances of different item ranking approaches.

\noindent \textbf{Parameter settings.} For the above baselines, we have carefully explored the corresponding parameters, i.e., the number of dimensions and regularization parameters. Besides, 
for SQL-Rank, we chose the ratio of subsampled unobserved items to positive items as $ 3:1 $ following the authors' guidance.
For MF-SetRank, 
we tuned the learning rate in $ [0.1,0.2,...,1.0] $ and the decay rate in $ [0.9, 0.93, 0.95, 0.97, 0.99] $. 
We also fixed the sampling ratio $ \tau $ to $3 $. Then we tuned the number of dimensions $ r $ in [50, 100, 150, 200, 250, 300] and the regularization parameter $ \lambda $ in $ [0.2,0.3,...,1.9,2.0] $. 

For Multi-VAE, we set the encoder as $ 2 $-layer MLP with dimensions $ 600\times 200 $ and decoder with dimensions $ 200\times 600 $. For the other four ``Deep'' methods, we fixed the user network as $ 2 $-layer MLP with dimensions $ 512\times 100 $ and item network with dimensions $ 1024\times 100 $. 
Then we performed Adam~\cite{kingma2014adam} algorithm for optimization and tune the learning rate from $ 0.0001 $ to $ 0.01 $.

\subsection{Overall Performance Comparison}
We present the overall recommendation performance results of the nine methods in Table~\ref{tab1} under two types of settings, i.e., $ P = 5 $ and $ P = 10 $, since the top recommended items are much more important in practical scenes. As shown in the results, Deep-SetRank achieves the best performance against all the baseline methods on every dataset. Specifically, Deep-SetRank outperforms the best baselines by an average relative boost of 4.28\% for the metric P@$ 5 $ on the four datasets. Besides, MF-SetRank achieves the best performance against all the other MF based baselines. Specifically, MF-SetRank outperforms the state-of-the-art MF based method, SQL-Rank, by an average relative boost of 11.57\% for the metric P@$ 5 $. We can also observe that NN based models have stronger embedding ability and can perform better than MF based models. Nevertheless, it is notable that MF-SetRank has achieved comparable performances with NN based methods, such as Multi-VAE and Deep-SQL. The outstanding performances clearly demonstrate the effectiveness of our setwise approaches. We can also observe that SetRank achieves the largest relative boost to the other baselines on the sparsest dataset, \emph{CiteULike}, which shows its superior capacity for handling sparsity problem. 
Another notable thing is that 
listwise approaches seem to perform better in top-$ 5 $ metrics than top-$ 10 $ metrics. This is probably because they pay more attention to the top ranks in an item list. On the opposite, SetRank treats every positive item or every unobserved item fairly thus can perform well in both top-$ 5 $ and top-$ 10 $ metrics.

\subsection{Hyper-parameter Investigations}\label{sectiont}

\noindent \textbf{Effectiveness of negative sampling.}
As mentioned in Section Implementation, it is unnecessary to utilize all the unobserved items for gradient calculations in SetRank. We can just randomly sample $ \tau \cdot J_i  $ negative items for each user $ i $ in each epoch. Since the number of positive items $ J_i $ is usually far smaller than the number of total items, there are few unobserved item overlaps for each user among different epochs. In this subsection, we fix all the other parameters to be the same and evaluate the influence of sampling ratio $ \tau $ on final recommendation results. The P@$ 5 $ results are shown in Figure~\ref{fig:tkind}. We find that when $ \tau=3 $, the performance is good enough. Even if we further enlarge the value of $ \tau $, the result would not increase significantly. 

\noindent \textbf{Sensitivities of latent factors.}
In this paper, we factorize the score matrix into the product of user and item latent factors in a low-rank space. Therefore, the rank $ r $ of latent space is quite influential to the result. If the rank $ r $ is too small, the model could not fit the real-world data well while if $ r $ is too larger, it may cause the overfitting problem. 
We varied $ r $ to train our method and then presented the results in Figure~\ref{fig:rkind}. We can observe that the performance result of SetRank is not good when $ r=50 $. With a larger value of $ r $, the performance of MF-SetRank tends to be much better. Thus, we suggest adopting a large value for $ r $ to get the best performance in MF-SetRank. 
By comparison, $ r=100 $ seems to be good enough for Deep-SetRank.

\section{Conclusion}
In this paper, we proposed a setwise Bayesian approach, namely SetRank, for collaborative ranking. SetRank has the ability in  accommodating the characteristic of implicit feedback in recommender systems. Specifically, we first designed a novel setwise preference structure. Then, we maximized the posterior probability of the setwise preference structure to complete the Bayesian inference. In particular, we designed two implementations, MF-SetRank and Deep-SetRank. Moreover, we provided the theoretical analysis of SetRank to show that the bound of excess risk can be proportional to $\sqrt{M/N}$. Finally, extensive experiments on four real-world datasets clearly validated the advantages of SetRank over various state-of-the-art baselines.

\section{Acknowledgments}
This work was supported by grants
from the National Natural Science Foundation of China (No.91746301,
61836013). 

\section{Appendix}

\subsection{Proof of Theorem 1}
\begin{proof}
	. Owing to the independence assumption, the following equation holds: 
	
	\begin{small}
		\begin{align}\label{equ1}
		\notag p(>_{total}|X) &= \prod_{i=1}^N p(>_{i}|X_i) \\
		&= \prod_{i=1}^N \prod_{j\in P_i} p(Y_{ij} \leq \min_{k\in O_i}\{Y_{ik}\}|X_i),
		\end{align}
	\end{small}
	
	\noindent where $ \min\limits_{k\in O_i}\{Y_{ik}\} $ obeys an exponential distribution with rate $ \sum_{k\in O_i}\phi(X_{ik}) $. Then we have
	
	\begin{small}
		\begin{align}\label{equ2}
		\notag p(Y_{ij} &\leq \min_{k\in O_i}\{Y_{ik}\}|X_i) \\
		\notag &= \int_0^\infty \phi(X_{ij}) e^{-u\phi(X_{ij})} e^{\sum_{k\in O_i}-u\phi(X_{ik})}du\\ 
		&= \frac{\phi(X_{ij})}{\phi(X_{ij})+ \sum_{k\in O_i}\phi(X_{ik})}.
		\end{align}
	\end{small}
	
	With Equation~\ref{equ1} and~\ref{equ2}, we obtain the conclusion.
\end{proof}

\subsection{Proof of Theorem 2 and  Theorem 3}
To prove Theorem 2, we first follow \cite{wu2018sql} to propose an important lemma to bound the excess risk by an empirical process term.

\begin{lemma}\label{lemma1}
	Supposing $ \hat{X} := \arg\min\limits_X -\log p(>_{total}|X) \ \text{such that}\ X \in \mathcal{X} $ and there is a $ X^*\in \mathcal{X}$ such that $ >_{total} $ is generated from $ p(>_{total}|X^*)$. Then we have the following inequality, where $ \mathbb{E} $ is for the draw of $ >_i $:
	
	\begin{small}
		\begin{equation}\label{equ:inequ}
		D(X^*,\hat{X})\leq -\frac{1}{N}\sum_{i=1}^N \left(\log \frac{p(>_i|X^*_i)}{p(>_i|\hat{X}_i)} - \mathbb{E} \log \frac{p(>_i|X^*_i)}{p(>_i|\hat{X}_i)}\right).
		\end{equation}
	\end{small}
	
\end{lemma}

\begin{proof}
	. Due to the optimality condition, we have
	
	\begin{small}
		\begin{equation}\label{equ3}
		\sum_{i=1}^N -\log p(>_i|\hat{X}_i) \leq \sum_{i=1}^N -\log p(>_i|X^*_i).
		\end{equation}
	\end{small}
	
	Actually, Equation~\ref{equ3} is equivalent to
	
	\begin{small}
		\begin{equation}\label{equ4}
		\notag \frac{1}{N}\sum_{i=1}^N \log \frac{p(>_i|X^*_i)}{p(>_i|\hat{X}_i)} \leq 0.
		\end{equation}
	\end{small}
	
	Thus, it is easy to obtain the conclusion.
\end{proof}

As we can see from Lemma~\ref{lemma1}, if we fix $ \hat{X} $, the empirical process term (the RHS of Equation~\ref{equ:inequ}) is a random function of the preference structure $ >_{total} $ with mean zero. However, $ \hat{X} $ is also random so that we have to uniformly bound the empirical process term over $ \hat{X}\in \mathcal{X}$. To apply Dudley’s chaining~\cite{talagrand2006generic}, we first bound the variations between two preference scores $ X_i $ and $ X_i' $ with Lemma~\ref{lemma2}: 

\begin{lemma}\label{lemma2}
	Define the difference function $ \Delta(>_i|X_i, X_i') := \log \frac{p(>_i|X_i)}{p(>_i|X_i')}$. If a single entry $ Y_{il} $ changes, it would cause the transformation of setwise preference structure, i.e., $ >_i $ would be converted into $ >_i' $. We can bound the variations of the difference function in the form of:
	
	\begin{small}
		\begin{align}
		\notag |\Delta(>_i|X_i, X_i')&-\Delta(>_i'|X_i, X_i')|\\ 
		&\leq C \|\log \phi(X_i) - \log \phi(X_i')\|_\infty,
		\end{align}
	\end{small}
	
	\noindent where $ C = 2 + 2e^{2\alpha}J/K $.
\end{lemma}

\begin{proof}
	. If the change of $ Y_{il} $ does not lead to the change of $ P_i $  and $ O_i $, there is no influence for preference structure, i.e., $ >_i = >_i' $ and $ |\Delta(>_i|X_i, X_i')-\Delta(>_i'|X_i, X_i')| = 0 $. 
	
	Otherwise, we assume that item $ j'\in P_i $ and $ k'\in O_i $ exchange their status with each other so that in the new preference structure $ >_i' $, we have $ j' \in O_i' $ and $ k' \in P_i' $. In the following part of the proof, for ease of the statement, we denote $ \lambda_l = \phi(X_{il}) $ and $ \Lambda_j = \lambda_j + \sum_{k\in O_i \setminus \{k'\}} \lambda_k $. $ \lambda_l' $ and $ \Lambda_j' $ are defined analogously with $ X' $.
	
	So, we have
	
	\begin{small}
		\begin{equation}
		\notag \Delta(>_i|X_i, X_i') = \sum_{j\in P_i}\left(\log \frac{\lambda_j}{\lambda_j'} - \log\frac{\Lambda_j + \lambda_{k'}}{\Lambda_j'+ \lambda_{k'}'}\right),
		\end{equation}
	\end{small}
	\begin{small}
		\begin{align}
		\notag \Delta(>_i'|X_i, X_i') = &\sum_{j\in P_i \setminus \{j'\}}\left(\log \frac{\lambda_j}{\lambda_j'} - \log\frac{\Lambda_j + \lambda_{j'}}{\Lambda_j'+ \lambda_{j'}'}\right) \\
		&\notag + \log \frac{\lambda_{k'}}{\lambda_{k'}'}- \log\frac{\Lambda_{j'} + \lambda_{k'}}{\Lambda_{j'}'+ \lambda_{k'}'},
		\end{align}
	\end{small}
	
	and thus,
	
	\begin{small}
		\begin{align}
		\notag |\Delta(>_i|&X_i, X_i')-\Delta(>_i'|X_i, X_i')| \\
		&\notag= \left|\log \frac{\lambda_{j'}}{\lambda_{j'}'} - \log \frac{\lambda_{k'}}{\lambda_{k'}'} \right.\\
		&\notag\left.+ \sum_{j\in P_i \setminus \{j'\}}\left(\log\frac{\Lambda_j + \lambda_{j'}}{\Lambda_j'+ \lambda_{j'}'} - \log\frac{\Lambda_j + \lambda_{k'}}{\Lambda_j'+ \lambda_{k'}'} \right)\right|\\
		\notag &\leq \left| \log \frac{\lambda_{j'}}{\lambda_{j'}'} - \log \frac{\lambda_{k'}}{\lambda_{k'}'}\right|\\
		&\notag + \sum_{j\in P_i \setminus \{j'\}}\left|\log\frac{\Lambda_j + \lambda_{j'}}{\Lambda_j'+ \lambda_{j'}'} - \log\frac{\Lambda_j + \lambda_{k'}}{\Lambda_j'+ \lambda_{k'}'} \right|.
		\end{align}
	\end{small}
	
	Notice that
	
	\begin{small}
		\begin{align}
		\notag & \left|\log\frac{\Lambda_j + \lambda_{j'}}{\Lambda_j'+ \lambda_{j'}'} - \log\frac{\Lambda_j + \lambda_{k'}}{\Lambda_j'+ \lambda_{k'}'} \right| \\
		\notag &\leq \left|\log\frac{\Lambda_j + \lambda_{j'}}{\Lambda_j'+ \lambda_{j'}'} - \log\frac{\Lambda_j}{\Lambda_j'} \right| + \left| \log\frac{\Lambda_j}{\Lambda_j'} -  \log\frac{\Lambda_j + \lambda_{k'}}{\Lambda_j'+ \lambda_{k'}'} \right|\\
		\notag &= \left|\log\frac{\Lambda_j + \lambda_{j'}}{\Lambda_j} - \log\frac{\Lambda_j'+ \lambda_{j'}'}{\Lambda_j'} \right|\\
		\notag &\qquad + \left| \log\frac{\Lambda_j + \lambda_{k'}}{\Lambda_j} -  \log\frac{\Lambda_j'+ \lambda_{k'}'}{\Lambda_j'} \right|.
		\end{align}
	\end{small}
	
	Hence, we let $ \delta =  \|\log \phi(X_i) - \log \phi(X_i')\|_\infty$, and have
	
	\begin{small}
		\begin{equation}
		\notag \left|\log\frac{\Lambda_j }{\Lambda_j'}\right| \leq \max_{l}\left|\log \frac{\lambda_l}{\lambda_l'} \right| \leq \delta.
		\end{equation}
	\end{small}
	
	Further, we assume $ \beta_j = \max \{\frac{\lambda_{j'}}{\Lambda_j}, \frac{\lambda_{j'}'}{\Lambda_j'} \} $, and then
	
	\begin{small}
		\begin{align}
		\notag& \left|\log\left(1+\frac{ \lambda_{j'}}{\Lambda_j}\right) - \log\left(1+\frac{ \lambda_{j'}'}{\Lambda_j'}\right) \right| \\
		\notag& \leq \left|\log(1+\beta_j) -  \log(1+e^{-2\delta}\beta_j)\right| \leq |1-e^{-2\delta} | \left|\beta_j \right|.
		\end{align}
	\end{small}
	
	Considering that we have $ \Lambda_j \ge Ke^{-\alpha} $, we can derive $ \beta_j \leq e^{2\alpha}/K $. Similar conclusion can be obtained for $ \left| \log\left(1+\frac{\lambda_{k'}}{\Lambda_j}\right) -  \log\left(1+\frac{\lambda_{k'}'}{\Lambda_j'}\right) \right| $.
	
	Synthesize the analysis above, we thus have
	
	\begin{small}
		\begin{align}
		\notag |\Delta(>_i|X_i, X_i')&-\Delta(>_i'|X_i, X_i')|\\
		\notag &\leq 2\delta + 2\sum_{j\in P_i \setminus \{j'\}} |1-e^{-2\delta} |e^{2\alpha}/K \\
		\notag &\leq \delta(2 + 2e^{2\alpha}J/K).
		\end{align}
	\end{small}
	
	Consequently, let $ C = 2 + 2e^{2\alpha}J/K $ and we can come to the conclusion. 
\end{proof}

\begin{proof}
	\emph{of Theorem 2.} The empirical process function is defined as 
	
	\begin{small}
		\begin{equation}
		\nonumber \rho_N(x) := -\frac{1}{N}\sum_{i=1}^N \left(\log \frac{p(>_i|X^*_i)}{p(>_i|\hat{X}_i)} - \mathbb{E} \log \frac{p(>_i|X^*_i)}{p(>_i|\hat{X}_i)}\right).
		\end{equation}
	\end{small}
	
	From Theorem 1, we know that $ \rho_N(x) $ is a function of $ N\times M $ independent exponential random variables. And from Lemma 2, we know that the change of preference structure caused by the change of a single entry $ Y_{il} $ is bounded. Specifically, the accumulative squares of bounds are
	
	\begin{small}
		\begin{equation}
		\nonumber \sum_{i=1}^N \sum_{l=1}^M C^2 \|\log \phi(X_i) - \log \phi(X_i')\|_\infty^2 = M C^2 \|Z - Z'\|_{\infty, 2}^2.
		\end{equation}
	\end{small}
	
	Then according to McDiarmid's inequality~\cite{mcdiarmid1989method}, we have
	
	\begin{small}
		\begin{equation}
		\nonumber p\left\{N\left(\rho_N(x) - \rho_N(x')\right) > \epsilon \right\} \leq \exp \left( \frac{-2\epsilon^2}{M C^2 \|Z - Z'\|_{\infty, 2}^2}\right).
		\end{equation}
	\end{small}	
	
	As a result, the stochastic process $ \{N\rho_N(X)|X \in \mathcal{X} \} $ is a subGaussian field with canonical distance $ d(X,X') = \sqrt{M}C\|Z - Z'\|_{\infty, 2} $. Following Dudley's chaining~\cite{talagrand2006generic}, we can get the conclusion.
\end{proof}

\begin{proof}
	\emph{of Theorem 3.} \citet{wu2018sql} have proved that $ g(\mathcal{Z}) \leq c' \sqrt{rN} $ in the personalized collaborative setting, where $ c' $ is an absolute constant. Thus we can conclude the proof immediately.
\end{proof}

\bibliographystyle{aaai}
\bibliography{AAAI-WangC.2616}

\end{document}